\documentclass[11pt,a4paper]{article}

\usepackage{url}
\usepackage{graphicx}
\usepackage{dcolumn}
\usepackage{bm}
\usepackage[all,cmtip]{xy}
\usepackage{comment}
\usepackage{color}
\usepackage[american]{isodate}
\usepackage{amsmath}
\usepackage{amsfonts}
\usepackage{amssymb}
\usepackage{amsthm}
\usepackage{authblk}

\renewcommand{\vec}[1]{{\bf #1}}
\providecommand{\keywords}[1]{\textbf{\textit{Index terms---}} #1}

\renewcommand{\sec}[1]{Sec.~(\ref{sec:#1})}
\newcommand{\fig}[1]{Fig.~(\ref{fig:#1})}
\newcommand{\eq}[1]{Eq.~(\ref{eq:#1})}

\renewcommand{\a}{\alpha}
\newcommand{\g}{\gamma}
\newcommand{\e}{\epsilon}
\renewcommand{\b}{\beta}
\renewcommand{\l}{\lambda}
\newcommand{\sg}{\sigma}

\renewcommand{\d}{\partial}
\renewcommand{\k}{\kappa}
\renewcommand{\t}{\theta}
\renewcommand{\Re}[1]{\mathfrak{Re} \lcom #1 \rcom}
\newcommand{\f}{\mathfrak{f}}
\newcommand{\R}{\mathbb{R}}

\newcommand{\lcom}{\left[}
\newcommand{\rcom}{\right]}
\newcommand{\set}[1]{\lbrace #1 \rbrace}

\newcommand{\ra}{\rightarrow}

\newcommand{\lgans}{\textquotedblleft}

\newcommand{\nn}{\nonumber}

\numberwithin{equation}{section}      
\theoremstyle{plain}
\newtheorem{thm}{Theorem}[section]  
\newtheorem{lem}[thm]{Lemma}

\theoremstyle{definition}
\newtheorem{defn}{Definition}[section]

\begin{document}
\title{Uncertainty Estimates in the Heston Model via Fisher Information}

\author{Oliver Pfante}
\author{Nils Bertschinger}
\affil{Systemic Risk Group, Frankfurt Institute for Advanced Studies, Routh-Mofang-Stra{\ss}e 1, 60483 Frankfurt am Main}

\date{\today}
\maketitle
\begin{abstract}
    We address the information content of European option prices about
    volatility in terms of the Fisher information matrix. We assume
    that observed option prices are centred on the theoretical price provided
    by Heston's model disturbed by additive Gaussian noise. We fit the likelihood 
    function on the components of the VIX, i.e., near- and next-term put and call 
    options on the S{\&}P 500 with more than 23 days and less than 37 days to expiration and 
    non-vanishing bid, and compute their Fisher information matrices from the
    Greeks in the Heston model. We find that 
    option prices allow reliable estimates of volatility with negligible uncertainty 
    as long as volatility is large enough. Interestingly, if volatility drops below 
   a critical value, inferences from option prices become impossible because Vega, 
   the derivative of a European option w.r.t. volatility, nearly vanishes.
\end{abstract}


\section{Introduction}

Volatility of stock processes is a highly volatile time-process itself. 
This insight led to the introduction of volatility
indices like the VIX (1993) and its off-springs, based on the work
\cite{Brenner1989,Brenner1993}, which make volatility a
trademark in its own right subject to similar stochastic movements as
stock prices. According to this view volatility seems to be responsible
for several statistical properties of observed stock price
processes. In particular, volatility clustering, i.e., large
fluctuations are commonly followed by other large fluctuations and
similarly for small changes \cite{Bouchaud2003}. Another feature is
that, in clear contrast with price changes which show negligible
autocorrelations, volatility autocorrelation is still significant for
time lags longer than one year
\cite{Perello2008,Muzy2000,Bouchaud2003,Lo1991,Ding1993,LeBaron2001}. Additionally,
there exists the so-called leverage effect, i.e., much shorter (few
weeks) negative cross-correlation between current price change and
future volatility
\cite{Bouchaud2001,Bouchaud2003,Black1973,Bollerslev2006}.\\

In stochastic volatility models, volatility is considered as a hidden process which
can only be observed indirectly via its effect on stock price
dynamics. Thus, in practice it has to be inferred from market data. In
a previous paper \cite{PfanBerti2016}, we have shown that daily stock returns provide only
very limited information about volatility. Thus, there are intrinsic
limits on how precisely volatility can be recovered from market
data. Here, we consider a related question when recovering volatility
from option price data. To our knowledge, there are no estimates of 
the uncertainty associated with volatility when inferred from option prices.\\
We tackle this question in terms of Fisher information which
quantifies the uncertainty of a maximum likelihood estimate by the
curvature of the likelihood function around its maximum. A shallow
maximum would have low information as the parameters are only weakly
determined; while a sharply peaked maximum would have high
information. Further, the famous Cramer-Rao bound links Fisher
information to the minimum variance of an unbiased estimator, thus
providing fundamental limits on the reliability of parameter
estimation. \\
In the present context, among the vast family of stochastic volatility
models, we focus on Heston's model. It successfully models statistical
properties of stock returns \cite{Christie1982,Bouchaud2003,Dragulescu2002,Mandelbrot1997}
and the implied volatility surface \cite{Forde2012,GatheralVolSur,Janek2011,Poon2009, Jacquier2013,Jacquier2016,Mazzon2015}. We refer also to the introduction
of \cite{roughHeston} for a discussion of the empirical properties of Heston's model. There,
the authors derive an analytical expression for the characteristic function of an extension
of the Heston model in terms of fractional Brownian motions. Most important, along with Heston's original paper \cite{Heston1993} came also an analytical expression of the European option prices in terms of characteristic functions and many algorithms have been proposed for their computation -- see also Rouah's book \cite{HestonExt} and the references therein for a thorough discussion of these algorithms.\\
We illustrate our results on an option portfolio which underlies the
VIX index and compute the absolute and relative uncertainties of S\&P
500 volatility over the last decade. Overall, we find that in contrast
to stock returns alone \cite{PfanBerti2016}, option prices provide substantial information
about volatility making inferred volatility a precise
estimator. Interestingly, very small volatilities are most difficult
to estimate as Vega almost vanishes below a critical value. This in
turn leads to huge relative errors in small inferred volatilities. The
VIX, being a variance swap on the average volatility over 30 days, is
much more stable with a relative uncertainty never exceeding 3\% over
the considered data set.\\

The paper is structured as follows: \sec{Heston} introduces Heston's
model and associated pricing formulas. Further, we explain how the
fractional fast Fourier transform allows an efficient computation of
the Heston Greeks. In \sec{FisherInfo} we state the formal definition
and important properties of Fisher information. \sec{Surface}
illustrates how the Fisher information varies w.r.t. strike and maturity
in the Heston model. Further, we discuss the role of Vega on the
reliability of volatility estimation. Finally, in \sec{VIX} we compute
the Fisher information for the volatility of the S\&P 500
index as well as the VIX index.

\section{Pricing options in Heston's Model}
\label{sec:Heston}

\subsection{Heston's Stochastic Volatility}
Introducing Heston's model for pricing options we follow \cite{HestonExt}. The Heston model assumes that the underlying stock price, $S_t$, follows a Black-Scholes-type stochastic process, but with a stochastic variance $v_t$ that follows a Cox, Ingersoll, and Ross process. Hence, the Heston model is represented by the bivariate system of stochastic differential equations (SDEs)
\begin{align} \label{eq:CoxIngersoll}
dS_t &= \mu S_t dt + \sqrt{v_t}S_t dW_{1,t} \nn \\
dv_t &= \k(\t - v_t) dt + \sg \sqrt{v_t} dW_{2,t} 
\end{align}
with the instantaneous correlation $dW_{1,t}dW_{2,t} = \rho dt$ of the two Brownian motions. The parameters of the model are
\begin{displaymath}
\begin{tabular}{ll}
\text{drift of the stock} & $\mu$ \\ 
\text{relaxation parameter} & $\k > 0$ \\ 
\text{long-term mean of the variance} & $\t > 0$ \\ 
\text{leverage-effect parameter} & $\rho \in [0,1]$ \\ 
\text{volatility of the variance} & $\sg > 0$\, . \\ 
\end{tabular} 
\end{displaymath}
Furthermore, the price at time $t$ of a zero-coupon bond paying 1\$ at maturity $t + \tau$ is 
\[
P(t, t+\tau) = e^{-\tau r}
\]
with constant interest rate $r$. Neglecting volatility risk premium, change of measure yields the log-price $x_t = \log S_t$ and variance $v_t$
\begin{align*}
d x_t &= \left(r - \frac{1}{2} v_t \right) dt + \sqrt{v_t} d \tilde W_{1,t} \\
d v_t &= \k(\t - v_t) dt + \sg \sqrt{v_t} d \tilde W_{2,t}
\end{align*}
w.r.t. to the risk neutral measure $\mathbb{Q}$, see \cite{HestonExt} or Heston's original work \cite{Heston1993} for a detailed derivation. If we include a continuous dividend yield $q$, the time drift of the log-price becomes $r - q - 1/2 v_t$.

\subsection{European options}
We present the price for a European call and put option in Heston's model in the formulation of Carr and Madan \cite{Carr1999}. Henceforth, we abbreviate the log-price $x_t = \log S_t$ at time $t$ simply by $x$ and similarly the variance $v_t$ at time $t$ by $v$. We only consider the characteristic function $f$ of the cumulative distribution $P^{\mathbb{Q}}(X_{\tau} > \log K)$ w.r.t. the risk-neutral probability, that is
\[
P^{\mathbb{Q}}(X_{\tau} > k) = \dfrac{1}{2} + \dfrac{1}{\pi} \int_0 ^{\infty} \Re{\dfrac{e^{- i \phi k}f(\phi, x, v, \tau)}{i\phi}} \, d\phi \, .
\]
with
\[
f(\phi, x, v, \tau) = e^{(C(\phi, \tau) + D(\phi, \tau)v + i\phi x)} \, 
\]
and the logarithmic strike $k = \log K$. The little Heston Trap formulation \cite{Albrecher2007} yields
\begin{align*}
C(\phi, \tau) &= i(r-q)\phi \tau + \dfrac{\k \t}{\sg^2} \lcom (Q -d)\tau - 2 \log \left( \dfrac{1 - c e^{-d \tau}}{1 - c}\right) \rcom \\
D(\phi, \tau) &= \dfrac{Q - d}{\sg^2}\left( \dfrac{1 - e^{-d\tau}}{1 - ce^{-d\tau}} \right)
\end{align*}
with 
\begin{align*}
c &= \dfrac{Q - d}{Q + d} \\
d &= \sqrt{Q^2 + \sg^2(i\phi + \phi^2)} \\
Q &= \k - i \rho \sg \phi \, .
\end{align*}
We introduce
\begin{equation} \label{eq:general}
E(\e, x, v, r, \k, \t, \rho, \sg, k, \tau) = \dfrac{e^{-\e \a k}}{\pi} \int_0^{\infty}\Re{e^{-ik\phi} \hat e(\e, \phi, x, v, \tau)} \, d\phi
\end{equation}
with
\[
\hat e(\e, \phi, x, v, \tau) = \dfrac{e^{-r \tau}f(\phi - i(\e \a + 1), x, v)}{(\e \a)^2 + \e \a - \phi^2 + i \phi (2\e \a + 1) }\, .
\]
where $\a > 0$ is a positive damping factor which can be chosen according to an optimization scheme outlined in \cite{Carr1999} and $\e \in \set{1, -1}$. We obtain the European call $C(x, v, r, \k, \t, \rho, \sg, k, \tau)$ and put $P(x, v, r, \k, \t, \rho, \sg, k, \tau)$ price, respectively, via
\begin{align}  \label{eq:carr}
C(x, v, r, \k, \t, \rho, \sg, k, \tau) &= E( 1, x, v, r, \k, \t, \rho, \sg, k, \tau) \nn \\
P(x, v, r, \k, \t, \rho, \sg, k, \tau) &= E(-1, x, v, r, \k, \t, \rho, \sg, k, \tau) \, .
\end{align}
We refer to \cite{Carr1999} or the third chapter in Rouah's book \cite{HestonExt} for a derivation of the formulae \eq{carr}. Furthermore, we have put-call parity, see \cite{HestonExt},
\begin{equation} \label{eq:putcallParity}
P(x, v, r, \k, \t, \rho, \sg, k, \tau) = C(x, v, r, \k, \t, \rho, \sg, k, \tau) + e^k e^{-r\tau} - e^x e^{-q\tau} \, .
\end{equation}
The Heston Greeks in terms of the Carr-Madan formulation read
\begin{align} \label{eq:greek}
\dfrac{ \partial }{\partial \g}E(\e, x, & v, r, \k, \t, \rho, \sg, k, \tau)  \nn \\
&= \dfrac{e^{-\e \a k}}{\pi} \int_0^{\infty}\Re{e^{-ik\phi} \f_\g(\phi - i(\e \a + 1), x, v, \tau)  \hat e (\e, \phi, x, v, \tau)} \, d\phi
\end{align}
where $\f_\g$ are different functions for $\g \in \set{ \sg_0,  \k, \sqrt{\t}, \rho, \sg}$ and $\sg_0 = \sqrt{v}$ denotes the volatility.

\subsection{Fractional Fourier Transform}
Since the call and put price in \eq{carr} is expressed via a single Fourier integral we can apply a Fractional Fast Fourier Transform to achieve 
a simultaneous computation of the call and put prices \eq{carr} for various strikes. We follow the outline in \cite{HestonExt}. We approximate the Fourier integral in \eq{general} by Simpson's integration scheme over a truncated integration domain for $\phi$, using $N$ equidistant points
\[
\phi_j = j \eta \qquad \text{for } j = 0, \ldots, N-1 
\]
where  $\eta$ is the increment. Simpson's rule approximates the integral \eq{greek} as
\begin{align} \label{eq:simpson}
\dfrac{\partial}{\partial \g} E(& \e, x,   v, r, \k, \t, \rho, \sg, k, \tau) \nn \\
& \approx  \dfrac{e^{-\e \a k}\eta}{\pi}  \sum_{j=0}^{N-1} \Re{e^{i\phi_j k}\f_\g(\phi_j - i(\e \a + 1), x, v, \tau) \hat e (\e, \phi_j, x,v, \tau)}w_j
\end{align}
with the weight $w_0 = w_{N-1} = 1/3$, $w_j = 4/3$ iff $j$ is odd and $w_j= 2/3$ otherwise.
Since we are interested in strikes near the money the range for the log-strikes $k$ needs to be centred
on the log-price $x$. The strike range is, thus, discretized using $N$ equidistant points
\[
k_u = -b + u\l + x \qquad \text{for } u = 0, \ldots, N-1
\]
where $\l$ is the increment and $b = N\l/2$. This produces log-strikes over the range $[\log S - b, \log S + b - \l]$. For a log-strike on the grid $k_u$ we can write the sum \eq{simpson} as 
\begin{align} \label{eq:FFRT}
&\dfrac{\partial}{\partial \g} E(k_u) \nn \\ 
 &\approx  \dfrac{e^{-\e \a k_u}\eta}{\pi}  \sum_{j=0}^{N-1} \Re{e^{i\phi_jk_u} \f_\g(\phi_j - i(\e \a + 1), x, v, \tau) \hat e (\e, \phi_j, x, v, \tau)}w_j  \nn \\
&= \dfrac{e^{-\e \a k_u}\eta}{\pi}  \sum_{j=0}^{N-1} \Re{e^{ij\eta(-b + u\l + x)} \f_\g(\phi_j - i(\e \a + 1), x, v, \tau) \hat e (\e, \phi_j, x, v, \tau)}w_j  \nn \\
&= \dfrac{e^{-\e \a k_u}\eta}{\pi}  \sum_{j=0}^{N-1} \Re{e^{i\eta\l uj}e^{i(b - x)\phi_j} \f_\g(\phi_j - i(\e \a + 1), x, v, \tau) \hat e (\e, \phi_j, x, v, \tau)}w_j 
\end{align}
for $u = 1, \ldots, N-1$. Applying a discrete Fast Fourier Transform (FFT) on \eq{FFRT} imposes the restriction 
\begin{equation} \label{eq:restriction}
\l \eta = \dfrac{2 \pi}{N}
\end{equation}
on the choice of the increments $\l$ and $\eta$ which entails a trade off between the grid sizes. Hence, Chourdakis \cite{Chourdakis2005} introduced the Fractional Fast Fourier Transform (FRFT) to relax this important limitation of the discrete FFT. The term $2 \pi/N$ in \eq{restriction} is replaced by a general term $\b$ and \eq{FFRT} becomes 
\begin{equation} \label{eq:xform}
\hat x_u  = E(k_u)\approx  \dfrac{e^{-\e \a k_u}\eta}{\pi}  \sum_{j=0}^{N-1} \Re{e^{i\beta uj}x_j}
\end{equation}
with 
\[
x_j = e^{i(b - x)\phi_j}\f_\g(\phi_j - i(\e \a + 1), x, v, \tau)  \hat{e} (\e, \phi_j, x, v, \tau)w_j \quad j = 1, \ldots, N-1 \, .
\]
The relationship between the grid sizes $\l$ and $\eta$ becomes $\l \eta = \b$. Thus, we can
choose the grid size parameters $\eta$ and $\l$ freely. To implement the FRFT on the vector $\vec x = (x_0, \ldots, x_{N-1})$ we first define vectors
\begin{align*}
\vec y &= \left( \left\lbrace e^{-i \pi j^2 \beta}x_j\right\rbrace_{j = 0}^{N-1}, \{ 0 \}^{N-1}_{j = 0} \right) \\
\vec z &= \left( \left\lbrace e^{i \pi j^2 \beta}\right\rbrace_{j = 0}^{N-1}, \left\lbrace e^{i\pi (N-j)^2 \beta}\right\rbrace_{j = 0}^{N-1}\right) \, .
\end{align*}
Next, a discrete FFT of the vectors $\vec y$ and $\vec z$ yields $\hat{\vec y} = D(\vec y)$ and $\hat{\vec z} = D(\vec z)$ and we define the $2N$ dimensional vector 
\[
\hat{\vec h} = \hat{\vec y} \odot \hat{\vec z} = (y_jz_j)_{j=0}^{2N-1} \, .
\]
where $\odot$ denotes pointwise multiplication.

Now, we apply the inverse FFT on $\hat{\vec h} $ and take the pointwise product of the result with the
vector 
\[
e = \left( \left\lbrace e^{-i\pi j^2 \beta}	\right\rbrace_{j = 0}^{N-1}, \{ 0 \}^{N-1}_{j = 0} \right)
\]
to obtain 
\[
\hat{\vec x}  = e \odot D^{-1}( \hat{\vec h} ) =   e \odot D^{-1}( \hat{\vec y} \odot \hat{\vec z}) = 
 e \odot D^{-1} \left( D(\vec y) \odot D(\vec z) \right) \, .
\]
If we truncate the last $N$ elements we obtain the desired vector $\vec x$ of \eq{xform} with length $N$. Hence, the FRFT maps a vector of length $N$ onto another vector of length $N$, even though it uses intermediate vectors of length $2N$.

\section{Fisher Information}
\label{sec:FisherInfo}

\subsection{Introduction}
The outline on the Fisher information and the Cram\'{e}r-Rao inequality follows \cite{Cover2006}. We begin with a few definitions. Let $\{f(x; \vec \Theta) \}, \, \vec \Theta = (\t_1, \ldots, \t_{m} ) \in \mathcal{P} \subset \R^m$, denote an indexed family of densities $f: \mathcal{X} \ra \R_{\geq 0}, \, \int f(x; \vec \t) dx =1$ for all $\vec \Theta \in \mathcal{P}$. Here $\mathcal{P}$ is called the parameter set.

\begin{defn}
An \textit{estimator} for $\vec \Theta$ for sample size $n$ is a function $T: \mathcal{X}^n \ra \mathcal{P}$.
\end{defn}

An estimator is meant to approximate the value of the parameter. It is therefore desirable to have some idea of the goodness of the approximation. We will call the difference $T - \vec \Theta$ the \text{error} of the estimator. The
error is a random variable. 

\begin{defn}
The \textit{bias} of an estimator $T(X_1, X_2, \ldots, X_n)$ for the parameter $\vec \Theta$ is the expected value of the error of the estimator, i.e., the bias is 
\begin{align*}
\mathbb{E}_{\vec \Theta} & \left[  T(x_1, x_2, \ldots, x_n) - \vec \Theta \right] \\
                  &= \int \left( T(x_1, x_2, \ldots, x_n) - \vec \Theta \right) f(x_1; \vec \Theta) \cdots f(x_n; \vec \Theta)  \, dx_1 \cdots dx_n \, .
\end{align*}
The estimator is said to be \textit{unbiased} if the bias is zero for all $\vec \Theta \in \mathcal{P}$ (i.e., the expected value of the estimator is equal to the parameter).
\end{defn}

The bias is the expected value of the error, and the fact that it is zero does not guarantee that the error is low with high probability. We need to look at some loss function of the error; the most commonly chosen loss function is the quadratic form
\begin{align} \label{eq:lossFunction}
\Sigma  &(T(X_1, X_2, \ldots,  X_n) ) = \nn \\
& \mathbb{E} \left[ (T(X_1, X_2, \ldots, X_n) - \vec \Theta)(T(X_1, X_2, \ldots, X_n) - \vec \Theta)^{T} \right] 
\end{align} 
which reduces to the covariance matrix for an unbiased estimator. Recall, covariance matrices are positive semi-definite and we write $A \leq B$ for two positive semi-definite $m \times m$ matrices if $x^{T} A x \leq x^{T} B x$ for all $x \in \R^{m}$. 

\begin{defn}
An estimator $T_1(X_1, X_2, \ldots, X_n)$ is said to \textit{dominate} another estimator $T_2(X_1, X_2, \ldots, X_n)$ if, for all $\vec \Theta$, 
\[
\Sigma(T(X_1, X_2, \ldots, X_n) ) \leq \Sigma(T(X_1, X_2, \ldots, X_n) )   \, .
\]
\end{defn}

The Cram\'{e}r-Rao lower bound gives the minimum quadratic loss of the best unbiased estimator of $\vec \Theta$.
First, we define the score gradient of the distribution $f(x; \vec \Theta)$.

\begin{defn}
The \textit{score gradient} $V $ is a random variable defined by
\[
V = \left( \dfrac{\partial}{\partial \t_1} \log f(X; \vec \Theta), \ldots,  \dfrac{\partial}{\partial \t_m} \log f(X; \vec \Theta) \right)
\]
where $X \sim f(x; \vec \Theta)$.
\end{defn}

The expectation of every entry $V_i$ of the score gradient is 
\begin{align*}
\mathbb{E}V_{i} &= \int \left[\dfrac{\partial}{\partial \t_i} \log f(x; \vec \Theta) \right] f(x; \vec \Theta) \, dx \\
&= \int \dfrac{\frac{\partial}{\partial \t_i} f(x; \vec \Theta)}{f(x; \vec \Theta)} f(x; \vec \Theta) \, dx \\
&= \int \dfrac{\partial}{\partial \t_i} f(x; \vec \Theta) \, dx \\
&= \dfrac{\partial}{\partial \t_i} \underbrace{\int f(x; \vec \Theta) \, dx}_{=1} \\
& = 0 \, .
\end{align*}
Therefore, the covariance matrix of the score gradient $V$ is $\mathbb{E} VV^{T}$. These entries have a special meaning.

\begin{defn}
The \textit{Fisher information matrix} $J(\vec \Theta)$ is the covariance matrix of the score gradient, i.e., 
\[
J(\vec \Theta)_{ij} = \mathbb{E} \left[ \dfrac{\partial}{\partial \theta_i} \log f(x; \vec \Theta)  \dfrac{\partial}{\partial \theta_j} \log f(x; \vec \Theta) \right]
\]
\end{defn}

The following properties of the Fisher information are crucial for the present paper. First, as covariance matrix, the Fisher information matrix is positive semi-definite. Second, information is additive: the information yielded by two independent experiments is the sum of the information from each experiment
separately:
\[
J_{X,Y}(\vec \Theta) = J_X(\vec \Theta) + J_Y(\vec \Theta)
\]
Third, the Fisher information depends of the parametrization of the problem: suppose $\vec \Theta$ and $\vec \Lambda$ are $m$-vectors which parametrize
the estimation problem, and suppose $\vec \Theta$ is a continuously differentiable function of $\vec \Lambda$, then 
\begin{equation} \label{eq:fisherTrafo}
J(\vec \Lambda) = D^{T}J(\vec \Theta(\vec \Lambda)) D
\end{equation}
where the $ij$-th entry of the $m \times m$ Jacobian $D$ is defined by 
\[
D_{ij} = \dfrac{\partial \theta_i}{\partial \lambda_j}
\]
and $D^{T}$ denotes the transpose of $D$. Finally, the significance of 	the Fisher information is shown in the following theorem.

\begin{thm} \label{CramerRao}
(Cram\'{e}r-Rao inequality) The covariance matrix of any unbiased estimator $T(X)$ of the parameter $\vec \Theta$ is lower bounded by
the reciprocal of the Fisher information:
\[
\Sigma(T(X) ) \geq J(\vec \Theta)^{-1} \, .
\]
\end{thm}

The Fisher information is therefore a measure for the amount of \lgans information" about $\vec \Theta$ that is present in the data. It gives a lower
bound on the error in estimating $\vec \Theta$ from the data. 

\subsection{Fisher Information in Pricing Options}
The Fisher information matrix $J(\t_{ij})_{ij \in \set{\sg_0, \k, \sqrt{\t}, \rho, \sg}}$ of a single option with price $e$ is  a way of measuring the amount of information that the observable option price $e$ carries about the unknown parameters $\sg_0, \k, \sqrt{\t}, \rho, \sg$. Recall, $\sg_0$ is the volatility, $\k$ the relaxation parameter of the CIR process \eq{CoxIngersoll}, $\t$ the long-term mean of the variance, $\sg$ the volatility of the diffusion process \eq{CoxIngersoll}, and $\rho$ the leverage parameter, i.e., the instantaneous correlation between the two Brownian motions in \eq{CoxIngersoll}.\\

While Heston's option price $E(\e, x, v, r, \k, \t, \rho, \sg, k,
\tau)$ is a function of the stated parameters, the observed price $e$
might deviate from its theoretical value. Following standard practice
in fitting the volatility smile, we consider the mean squared loss
between the actually observed option price $e$ and its theoretical
value. In statistical terms, this corresponds to a noise model where
the option price $e$ is normally distributed around its theoretical
mean value with variance $\hat v$. The probability function for $e$,
which is also the likelihood function for the parameter vector $\vec
\Theta = (\sg_0, \k, \sqrt{\t}, \rho, \sg)$, is a function $f(e ; \vec
\Theta)$; it is the probability mass (or probability density) of the
random option price $e$ conditional on the value of $\vec \Theta$. The
$i, j$-th entry of the Fisher information matrix is defined as
\begin{equation} \label{eq:FisherInfoDef}
  J(\vec \Theta)_{ij} = \mathbb{E} \left[ \dfrac{\partial}{\partial \theta_i} \log f(x; \vec \Theta)  \dfrac{\partial}{\partial \theta_j} \log f(x; \vec \Theta) \right]
\end{equation}
with $ (\t_1,  \t_2, \ldots, \t_5) = \vec \Theta $. For a certain actual stock price $S$, strike $K$, and maturity $\tau$ we assume
\[
f(e | \vec \Theta) = \dfrac{1}{\sqrt{2 \pi \hat v}} \exp\left({-\dfrac{(e - E(\e, x, v, r, \k, \t, \rho, \sg, k, \tau))^2}{2 \hat v}} \right)
\]
where $x = \log S$, $k = \log K$, and $v = \sg_0^2$. We obtain the $ij$-th entry of the Fisher information matrix \eq{FisherInfoDef}
\begin{align}
  J(\vec \Theta)_{ij} &= \mathbb{E} \left[ \dfrac{\partial}{\partial \theta_i} \left(- \frac{1}{2} \log(2 \hat v) + \dfrac{(e - E(x,  \vec \Theta ,  k, \tau))^2}{2 \hat v}\right) \right. \times \nn \\
  & \qquad \qquad \qquad \left. \dfrac{\partial}{\partial \theta_j} \left(- \frac{1}{2} \log(2 \hat v) + \dfrac{(e - E(x,  \vec \Theta ,  k, \tau))^2}{2 \hat v}\right) \right] \nn \\
  &= \mathbb{E} \left[ \dfrac{e - E(x,  \vec \Theta ,  k, \tau)}{\hat v} \dfrac{\partial}{\partial \theta_i} E(x,  \vec \Theta ,  k, \tau) \right. \times \nn \\ 
  & \qquad \qquad \qquad \left. \dfrac{e - E(x,  \vec \Theta ,  k, \tau)}{\hat v} \dfrac{\partial}{\partial \theta_j} E(x,  \vec \Theta ,  k, \tau) \right] \nn \\
  &= \dfrac{1}{\hat v^2} \mathbb{E} \left[ (e - E(x,  \vec \Theta ,  k, \tau))^2 \right] \dfrac{\d}{\d \t_i} E(\e, x,  \vec \Theta , k, \tau) \dfrac{\d}{\d \t_j}  E(\e, x,  \vec \Theta , k, \tau) \nn \\
  &= \dfrac{1}{\hat v} \dfrac{\d}{\d \t_i} E(\e, x,  \vec \Theta , k, \tau) \dfrac{\d}{\d \t_j}  E(\e, x,  \vec \Theta , k, \tau) \nn \\
  &=  \dfrac{1}{\hat v} \nabla_{\Theta} E \left( \nabla_{\Theta} E \right)^T \, \label{eq:outerProduct}
\end{align}
where 
\[
\nabla_{\vec \Theta} E = \left( \dfrac{\d}{\d \t_i} E(\e, x, v, r, \k, \t, \rho, \sg, k, \tau) \right)_{i=1, \ldots, 5} 
\] 
is the gradient of the option price w.r.t. the parameters $\set{\sg_0, \k, \sqrt{\t}, \rho, \sg}$. \\
Furthermore, we consider deviations of observed prices of options with different strikes and maturities from their respective theoretical prices as independent. Additivity of Fisher information yields
\[
J(\vec \Theta) = \dfrac{1}{\hat v}\sum_{\tau \in \mathcal{T}, k \in \mathcal{K}} J_{\tau, k} (\vec \Theta)
\]
for simultaneously observed prices $\set{E(\e, x, v, r, \k, \t, \rho, \sg, k, \tau): \, \tau \in \mathcal{T}, k \in \mathcal{K} }$ of options with different maturities $\tau$ and log-strikes $k$.

\section{Inferring hidden parameters from a single European option}
\label{sec:Surface}

According to \eq{outerProduct}, Fisher information matrix of a European option with Gaussian Noise centred on the theoretical price \eq{general} provided by Heston's model is entirely determined by the first-order derivative of the option Price w.r.t. to the volatility $\sg_0$, and the parameters $\k, \t, \sg$ and $\rho$, respectively. We study these derivatives: first, their dependency of strikes and maturities; second, we have a closer look on Vega, i.e., the derivative of the option price w.r.t. to volatility $\sg_0$, observing a drop of its value for small volatilities.

\subsection{Greek-Surfaces}

The inverse of the Fisher information matrix $J(\vec \Theta)$ in \eq{outerProduct} provides a lower bound on the covariance matrix of an unbiased estimator $T(E)$ of the parameter $\vec \Theta = (\t_1, \ldots, \t_5) = (\sg_0, \k, \sqrt{\t}, \sg, \rho)$ from a single option price $E$. Hence, the diagonal elements $J(\vec \Theta)^{-1}_{ii}$, with $i = 1,2,\ldots, 5$, of the inverse $J(\vec \Theta)^{-1}$ yield lower bounds of the variances of the estimates $T(E)_i$ of the parameters $\t_i$ from a single observed option price $E$. The diagonal entries $J(\vec \Theta)^{-1}_{ii}$ can be lower bounded by the respective Fisher information $J(\t_i)^{-1}$. Note that $J(\t_i)$ is the information obtained when estimating the parameter $\t_i$ alone, i.e., considering all the other parameters as fixed. With this interpretation in mind the lemma below states that the variance of a joint estimator for all parameters $\vec \Theta$ simultaneously is larger than estimating each parameter $\t_i$ alone.

\begin{lem} \label{singleFit}
We have 
\[
 J^{-1}(\vec \Theta)_{ii} \geq J(\t_i)^{-1} 
\]
for $i = 1, \ldots, 5$.
\end{lem}

\begin{proof}
Let $\vec A$ be an invertible $n \times n$ matrix and let denote $\vec A_{ij}$ the $ij$-th block $\vec A$ for $i,j \in \set{1,2}$, i.e., 
\[ \vec A = \left[
\begin{array}{cc}
\vec A_{11} & \vec A_{12} \\ 
\vec A_{21} & \vec A_{22}
\end{array} \right] \, .
\]
Then, the inverse of $\vec A$ can be expressed as, by the use of 
\begin{align*}
\vec C_1 &= \vec A_{11} - \vec A_{12} \vec A_{22}^{-1} \vec A_{21} \\
\vec C_2 &= \vec A_{22} - \vec A_{21} \vec A_{11}^{-1} \vec A_{12} \, ,
\end{align*} 
as
\[
\vec A^{-1} = \left[ \begin{array}{cc}
\vec A_{11} & \vec A_{12} \\ 
\vec A_{21} & \vec A_{22}
\end{array} \right]^{-1} = \left[ \begin{array}{cc}
\vec C_1^{-1} & - \vec A_{11}^{-1}\vec A_{12} \vec C_2^{-1} \\ 
- \vec C_2^{-1} \vec A_{21} \vec A_{11}^{-1} & \vec C_2^{-1}
\end{array} \right] \, ,
\]
see equation $399$ in \cite{MatrixCook}. Assume $\vec A = J(\vec \Theta)$ and 
\[
\vec A_{11} = J(\vec \Theta)_{11} = \dfrac{1}{\hat v} \left( \dfrac{\partial}{\partial \t_1} E(x, k, \Theta, \tau)\right)^2 \, .
\] 
Then 
\[
J(\vec \Theta)^{-1}_{11} = \left( J(\vec \Theta)_{11} - \vec A_{12} \vec A_{22} \vec A_{21} \right)^{-1}
\]
$J(\vec \Theta)$ is positive semi-definite. This implies $\vec A_{12}^{T} = \vec A_{21}$, $\vec A_{22}$ is positive semi-definite as well and therefore $\vec A_{12} \vec A_{22} \vec A_{21} \geq 0$. This yields the inequality for $i=1$. Relabelling of the parameters yields the inequalities for the cases $i = 2,3,4, 5$ as well.
\end{proof}

According to \eq{outerProduct} we have
\[
J^{-1}(\vec \Theta)_{ii} \geq \hat v  \left (\dfrac{\partial}{\partial \t_i} E \right)^{-2} \quad \text{for } i = 1, \ldots, 5  \, .
\]
Hence, the gradient $\nabla_{\vec \Theta} E$ of an option price \eq{general} does not only entirely determine $J(\vec \Theta)$ but its components also provide first estimates on the uncertainty left about the parameters $\t_i$ from an estimate $T(E)$ derived from a single observation. \\
We computed the gradient for European call options for various strikes and maturities via the FRFT described in the previous section. We implemented the FRFT in \textsc{Haskell} for all our simulations. The choice of parameter values is consistent with the S{\&}P 500 whose parameter set was estimated in \cite{Ait-Sahalia2007} with $\k = 5.07$, $\t = 0.0457$, $\rho = -0.767$, $\sg=0.48$, even though our study, according to its qualitative character, holds true for any reasonable choice of parameters. Furthermore, the values of the time dependent parameters $r, q, v$ and $x$ are in accordance with the data for the  S{\&}P 500 on March 3, 2014. The continuously-compounded zero-coupon interest rate is $r  = 0.167{\%}$ with dividend yield $q = 1.894{\%}$. We assumed the S{\&}P 500 is quoted with $1845.73$ and a volatility $v = 0.0108$ which was fitted on call prices. All data is obtained from the \textsc{Option Metrics} database. Furthermore, the parameters of the FRFT were fixed as $N = 2^{11}$, $\eta = 0.4$, and $\l = 3.6549\text{e}-4$ which yields strike increments of approximately $0.68$ within in a range of $1269.5$ till $2683.5$. The damping factor is $\alpha = 1.5$ throughout the paper.\\

Throughout this section, we do not consider the variance $\hat v$ of the Gaussian Noise. Nevertheless, the figures allow for a comparison of the relative uncertainty of parameter estimates as $\hat v$ just corresponds to a global scaling of the Fisher information. Comparably no information from observing a single call price is gained about the relaxation parameter $\k$, the leverage parameter $\rho$, and the variance of the variance $\sg$. The estimates of the volatility $\sg_0$ and the long-term mean $\sqrt{\t}$ of the volatility are about $100$ times more precise. Interestingly, estimates of $\sg_0$ from call data have an optimal time scale: Vega, the derivative of the call price w.r.t. the volatility attains a global maximum for at the money call options with approximately one month maturity. In general estimates for all parameters are best for at the money call options and uncertainty increases with the distance of the strike from the quoted price of the underlying asset. Due to put-call Parity all these results hold true for put options as well. Hence, in the interest of space their detailed exposure is skipped.

\begin{figure}
\includegraphics[width=0.9\linewidth]{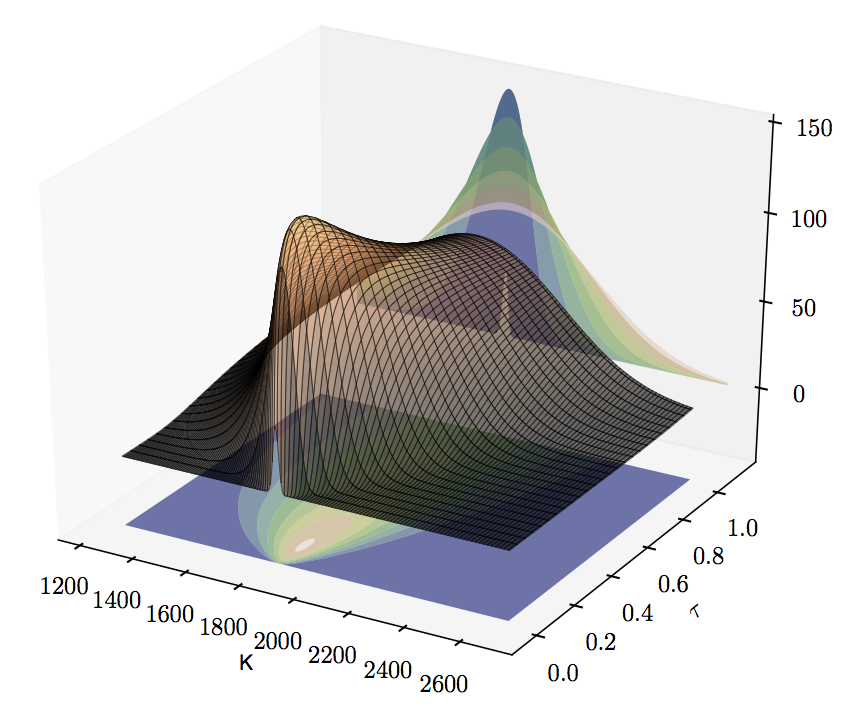}
\caption{The derivative of Heston's call Price w.r.t. the volatility $\sg_0$}
\end{figure}
\begin{figure}
\includegraphics[width=0.9\linewidth]{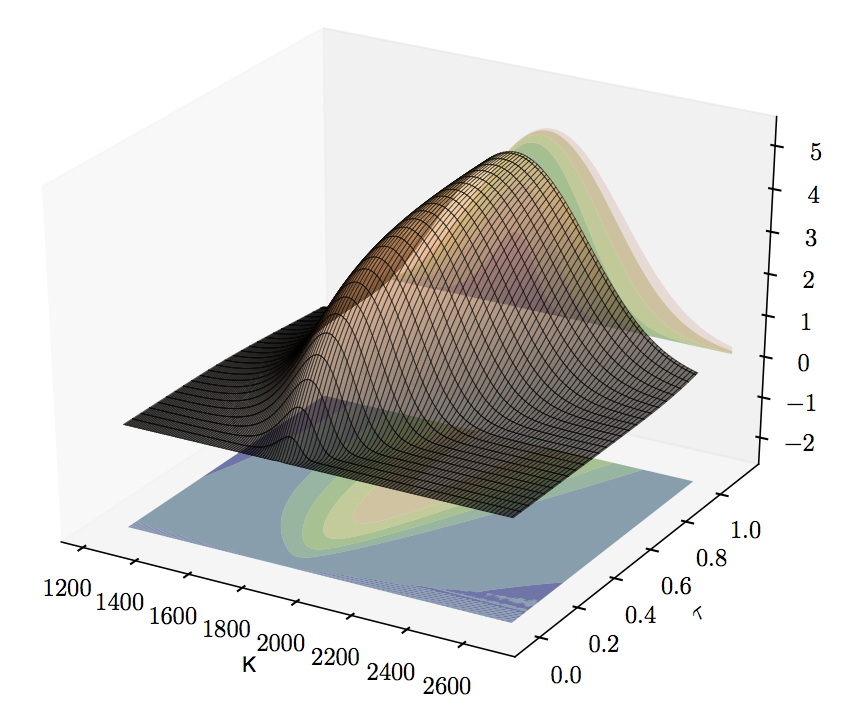}
\includegraphics[width=0.9\linewidth]{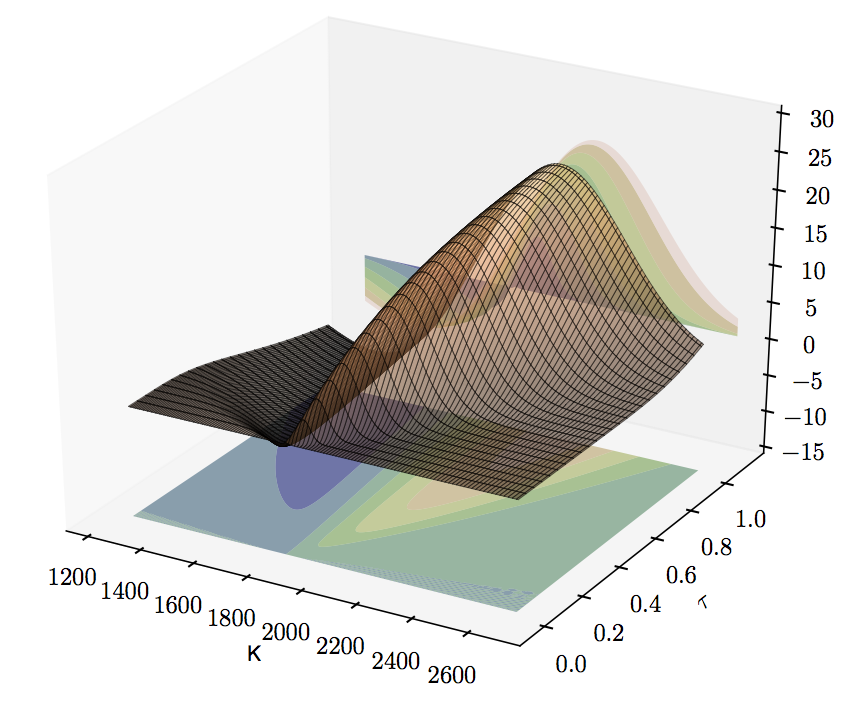}
\caption{The derivatives of Heston's call Price w.r.t. the relaxation parameter $\k$, upper figure, and the 
leverage parameter $\rho$.}
\end{figure}
\begin{figure}
\includegraphics[width=0.9\linewidth]{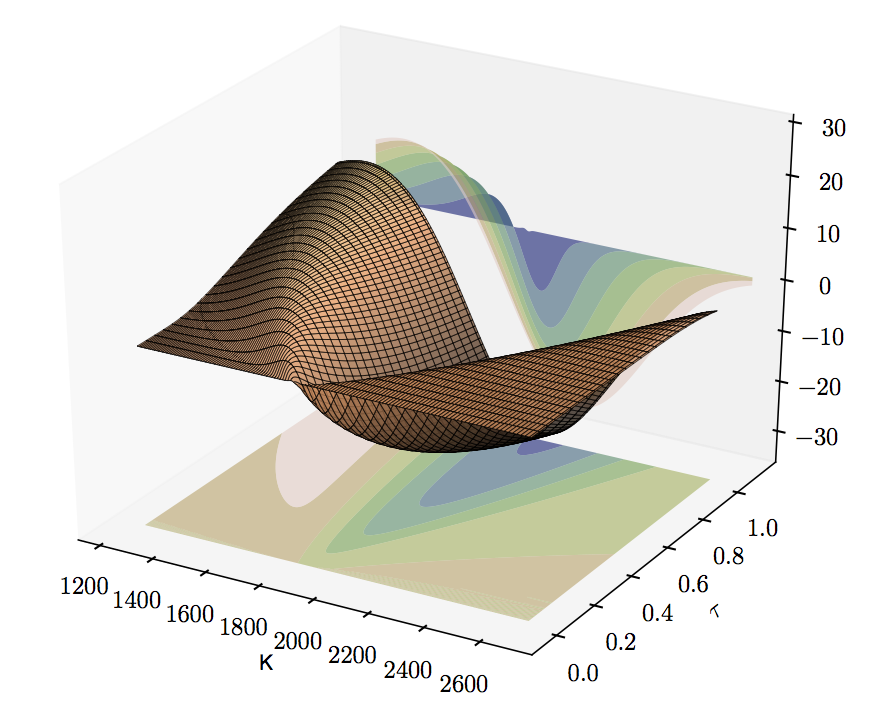}
\includegraphics[width=0.9\linewidth]{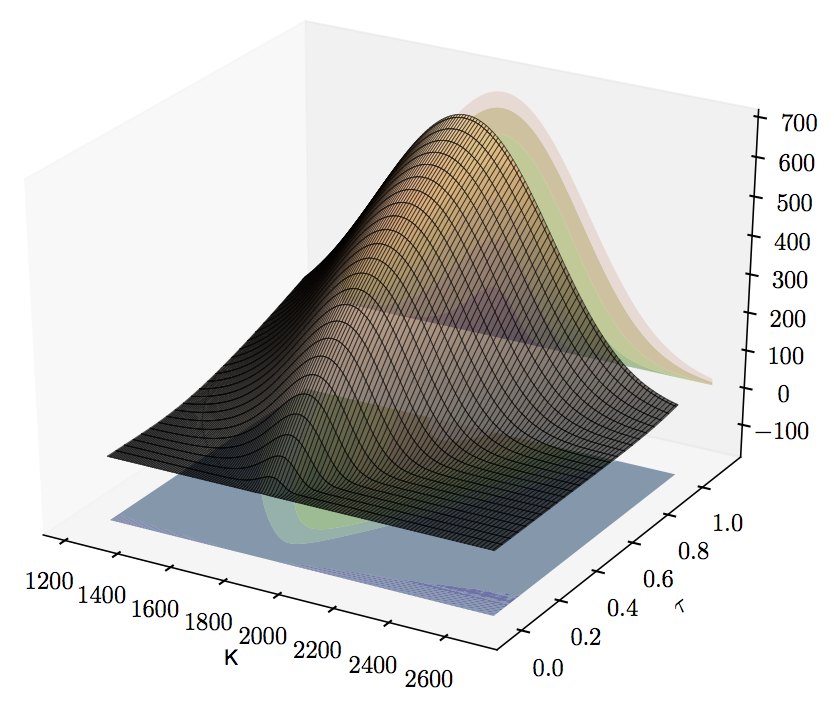}
\caption{The derivatives of Heston's call Price w.r.t. the volatility of the variance process $\sg$ and the
long-term mean of the volatility $\sqrt{\t}$.}
\end{figure}

\subsection{Vega-Drop}

In general parameters are not estimated from a single option price but for various options over an entire time-period. An estimate over a time-period $\mathcal T$ changes the picture in a twofold way.  First, not only maturity and strikes vary but also volatility $\sg_{t}$ and the price of the underlying asset $S_t$. Second, since volatility $\sg_t$ needs to be estimated for every day $t \in \mathcal T$ the parameter vector $\vec \Theta$ is no longer $(\sg_0, \k, \sqrt{\t}, \sg, \rho)$ but 
\begin{equation} \label{eq:paramVec}
(\sg_t)_{t  \in \mathcal{T}} \oplus (\k, \sqrt{\t}, \sg, \rho)
\end{equation}
where $\oplus$ denotes concatenation. Consequently, the Fisher information matrix $J(\vec \Theta)$ is an $(m +4) \times (m + 4) $ matrix where $m = |\mathcal{T}|$ is the number of days we consider. The first $m$ diagonal elements of $J(\vec \Theta)^{-1}$ provide lower error bounds on the estimates of volatilities $(\sg_t)_{\mathcal{T}}$. According to lemma \ref{singleFit} these diagonal elements are lower bounded by
\begin{equation} \label{eq:lower}
\hat v \left( \dfrac{\partial}{\partial \sg_t} E_t \right)^{-2} \quad \text{for } t \in \mathcal{T} \, 
\end{equation}
where $E_t$ is an observed option price at time $t$. Hence, Vega, i.e., the first-order derivative of the option price w.r.t. volatility, plays a crucial role for error estimates for the majority of the parameters in \eq{paramVec}. \\

\begin{figure}[ht] 
\includegraphics[width=1.0\linewidth]{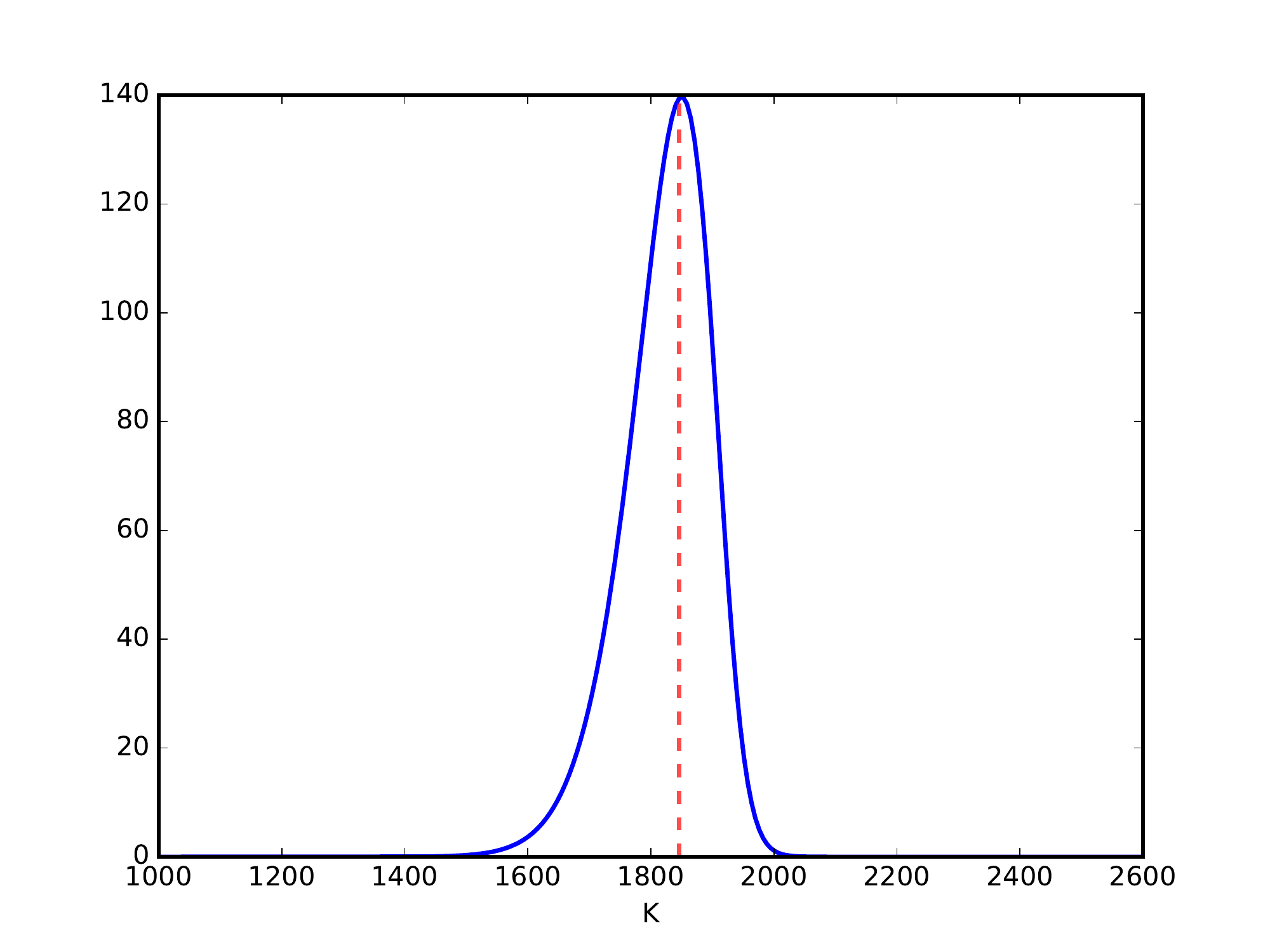}
\caption{\label{fig:vega}Vega for a European option priced by Heston's model over various strikes. It is strictly positive and has a global maximum near the spot price of the underlying asset (indicated by the red dashed line).}
\end{figure}

We plotted Vega of a European option traded on the S{\&}P $500$ in \fig{vega} over the same parameter set we have already applied for the explicit computation of the Greek-surfaces of the previous subsection: spot price $S= 1845.73 $, $r = 0.167 {\%}$, $q = 1.894 {\%}$, $\k = 5.07$, $\theta= 0.0457$, $\sg = 0.48$, $\rho =-0.767$, variance $v = 0.0108$, and maturity $\tau = 30$ days. As in the case of the classical Black Scholes Model, Vega is always positive, i.e., the value of an option increases with volatility, and Vega attains a maximum near the spot price of the underlying. Of greater importance is the dependency of Vega on the variance $v$ if the strike is fixed. 

\begin{figure}[ht]
\includegraphics[width=1.0\linewidth]{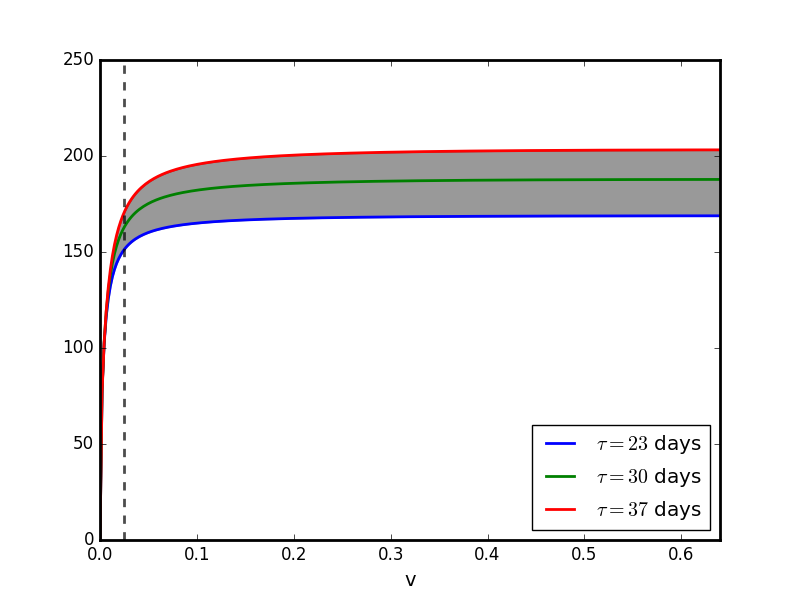}
\caption{ \label{fig:vegaDrop} Vega for a Eurpean option price by Heston's model over various variances for a fixed strike $K = 1845.73$ at the spot price of the underlying. Vega drops for small variances. The dashed line is at $v = 0.025$.}
\end{figure}

\fig{vegaDrop} plots Vega in Heston's model for an at-the-money option over various variances for three different maturities: $23$, $30$, and $37$ days; the range of maturities of the components of the VIX which is studied in the subsequent section. The curves for different maturities $\tau \in [23, 37]$ are located in the shaded region. One observes a dramatic drop of Vega for small variances $v$ which is indicated by the dashed line located at $v = 0.025$. According to \eq{lower} this fact sheds light on variance, and volatility estimates, respectively, from European options: if Vega decreases dramatically the error in estimating the volatility from options becomes large. In the subsequent section we show that this observed Vega-Drop is not only of theoretical but also practical interest if we investigate the quality of the VIX as a volatility proxy.

\section{Fisher Information and the VIX}
\label{sec:VIX}

The S{\&}P 500 is a major stock index which is calculated using the prices of approximately 500 component stocks of the biggest companies in the US. The S{\&}P 500, like other indices, employs rules that govern the selection of the component securities and a formula to calculate index values. The VIX index measures 30-day expected volatility of the S{\&}P 500 index and is comprised by options rather than stocks, with the price of each option reflecting the market's expectation of future volatility. Like conventional indices, the VIX calculation employs rules for selecting component options and a formula to calculate index values.  Roughly, the selection filters near- and next-term put and call options with more than 23 days and less than 37 days to expiration and non-vanishing bid -- see \cite{WhitePaper} for further details. With the aid of the \textsc{option Metrics} database we are able to replicate the part of the portfolio\footnote{The new VIX is computed based on S\&P 500 options which strike dates are renewed in monthly intervals. In addition, weekly options, traded under a different ticker, are used to cover the remaining weeks. These later options are not available in our data base.} of options entering the calculation of the VIX from October 16th 2003 until August 31th 2015. Assuming that all prices of the VIX components are centred on the theoretical price \eq{carr} in Heston's model disturbed by additive Gaussian Noise we address two questions to the data: first, how reliably can volatility be inferred from the VIX components; second, does the VIX really measure 30-day expected volatility if one considers the VIX index as a 30-day variance swap traded on the S{\&}P 500? We tackle both questions in terms of Fisher information.

\subsection{Inferring volatility}

We denote with $\mathcal{E}_t = \mathcal{C}_t \cup \mathcal{P}_t$ the set of all components of the VIX at day $t$ where $\mathcal{C}_t$ denotes the set of call options and $\mathcal{P}_t$ the set of put options, respectively. We introduce the set $\mathcal{T} = \set{16/10/2003, 17/10/2003, \ldots, 31/08/2015}$ of all trading days from October 16th 2003 until August 31th 2015. We assume that the price $e_t$ of an option $e$ in the set $\mathcal{E}_t$ is normally distributed 
\begin{align*}
p(e_t|\e, x_t, & v_t, r_t, \k, \t, \rho, \sg, k_e, \tau)  \\
&= \dfrac{1}{ \sqrt{2\pi \hat v}} \exp \left(  \dfrac{(e_t - E(\e, x_t, v_t, r_t, \k, \t, \rho, \sg, k_e, \tau))}{2 \hat v} \right)
\end{align*} 
with a fixed variance $\hat v$ and mean $E(\e, x_t, v_t, r_t, \k_e , \t, \rho, \sg, k, \tau)$ provided by \eq{general} where $x_t$ is the log-closing-price of the S{\&}P 500 at day $t$, $r_t$ the zero Coupon Bond yield at day $t$ with maturity $\tau$, and $k_e$ the log-strike of the option. Beside the option prices, also the daily zero Coupon Bond yield $r$ and the closing prices of the S{\&}P 500 are provided by the \textsc{Option Metrics} database. Furthermore we adopt the estimate in \cite{Ait-Sahalia2007} for the parameters of the stochastic volatility process \eq{CoxIngersoll}.
\begin{equation} \label{eq:stochParam}
\k = 5.07, \, \t = 0.0457, \sg = 0.48, \rho = -0.767
\end{equation}
We fit the volatility with daily time-resolution on call data via a maximum likelihood estimate on the call option prices. That is, for every day $t \in \mathcal{T}$, the variance $v_t$, and therefore volatility $\sg_t = \sqrt{v_t}$, is determined as the maximum of the joint, negative log-likelihood function
\[
- \sum_{c \in \mathcal{C}_t} \log \left(  p(c_t |  x_t, v_t, r_t, \k, \t, \rho, \sg, k_c, \tau) \right)
\]
which is equivalent to minimize the square-error
\[
\sum_{c \in \mathcal{C}_t} || c_t - C( x_t, v_t, r_t, \k, \t, \rho, \sg, k_c, \tau) ||^2 \, .
\]
This procedure yields a time series $(v_t)_{t \in \mathcal{T}}$ for the variance. Besides, the variance $\hat v$ is obtained from maximizing the negative, joint log-likelihood function
\[
\sum_{t \in \mathcal{T}} \sum_{c \in \mathcal{C}_t} \log \left(  p(c_t | x_t, v_t, r_t, \k, \t, \rho, \sg, k_c, \tau) \right)
\]
which yields the expression
\[
\hat v = \dfrac{1}{N} \sum_{t \in \mathcal{T}} \sum_{c \in \mathcal{C}_t} || c_t - C( x, v_t, r, \k, \t, \rho, \sg, k_c, \tau) ||^2
\]
where $N$ is the cardinality of the union $\bigcup_{t \in \mathcal{T}} \mathcal{C}_t$. For our data set of call options from October 16th 2003 until August 31th 2015 we obtain the value
\[
\hat v = 0.2952 \, .
\]
We have assembled all necessary ingredients to tackle the following thought problem: how much information is present in the VIX components about the unknown volatility process $(\sg_t)_{t \in \mathcal{T}}$, with $\sg_t = \sqrt{v_t}$, and the parameters $ \k, \sqrt{\t}, \sg$ and $\rho$? More precisely, Suppose the parameter vector 
\begin{equation} \label{eq:param}
\vec \Theta = (\sg_t)_{t \in \mathcal{T}}\oplus (\k, \sqrt{\t}, \sg, \rho)
\end{equation} 
and the joint log-likelihood function
\[
\sum_{t \in \mathcal{T}} \sum_{e \in \mathcal{E}_t} \log \left(  p(e_t | \e, x_t, v_t, r_t, \k, \t, \rho, \sg, k_e, \tau) \right) \, .
\]
What is the lower bound on the loss function \eq{lossFunction} of an estimator of $\vec \Theta$ provided all VIX components $\bigcup_{t \in \mathcal{T}} \mathcal{E}_t$ from  October 16th 2003 until August 31th 2015 as data? According to the Cram\'{e}r-Rao inequality the inverse of the Fisher information matrix $J(\vec \Theta)$ yields the answer. $J(\vec \Theta)$ adopts in the present context the block structure
\begin{equation} \label{eq:fisherVIX}
J(\vec \Theta)  =\dfrac{1}{\hat v} \left( \begin{array}{cc}
\vec A_{11} & \vec A_{12} \\ 
\vec A_{21} & \vec A_{22}
\end{array} \right) \, .
\end{equation}
$\vec A_{11} = (a_{ij})$  is an $m \times m$ diagonal matrix, where $m$ is the cardinality of $\mathcal{T}$ s.t.
\[
a_{ii} = \left(\partial_{\sg_{t_i}} E_i  \right)^2 =  4 v_{t_i} \left(  \partial_{v_{t_i}}  E_i \right)^2 \, 	
\]
where we write 
\begin{equation} \label{eq:Ei}
E_i = \sum_{e \in \mathcal{E}_{t_i}} E(\e, x_{t_i}, v_{t_i}, r_{t_i}, \k, \t, \rho, \sg, k_e, \tau) \, .
\end{equation}
$\vec A_{12}$ is an $m \times 4$ matrix,
\[
\vec A_{12} = \left( \begin{array}{cccc}
\partial_{\sg_{t_0}} E_0  \partial_{\k} E_0 & \partial_{\sg_{t_0}} E_0  \partial_{\sqrt{\t}} E_0 & \partial_{\sg_{t_0}} E_0  \partial_{\sg} E_0 & \partial_{\sg_{t_0}} E_0  \partial_{\rho} E_0 \\ 
\vdots & \vdots & \vdots & \vdots \\ 
\partial_{\sg_{t_m}} E_m  \partial_{\k} E_m & \partial_{\sg_{t_m}} E_m  \partial_{\sqrt{\t}} E_m & \partial_{\sg_{t_m}} E_m  \partial_{\sg} E_m & \partial_{\sg_{t_m}} E_m  \partial_{\rho} E_m
\end{array}  \right) \, , 
\]
and $\vec A_{21} = \vec A_{12}^{T}$. Finally, $\vec  A_{22}$ is the $4 \times 4$ matrix 
\[
\vec  A_{22} = \sum_{i = 1}^{m} \nabla E_i \nabla E_i^{T}
\]
where $\nabla E_i = (\partial_{\k} E_i, \partial_{\t} E_i, \partial_{\sg} E_i, \partial_{\rho} E_i)^{T}$ denotes the column gradient vector of $E_i$ w.r.t. the parameters of the stochastic process \eq{CoxIngersoll}. If we assume that the previously fitted variance time-series $(v_t)_{t \in \mathcal{T}}$ and the parameter set \eq{stochParam} provide an estimator of $\vec \Theta$, we can evaluate all the derivatives for computing the Fisher information matrix $J(\vec \Theta)$ \eq{fisherVIX} whose inverse yields a lower bound on the loss function of this estimator.
\begin{figure}[ht]
\includegraphics[width=1.0\linewidth]{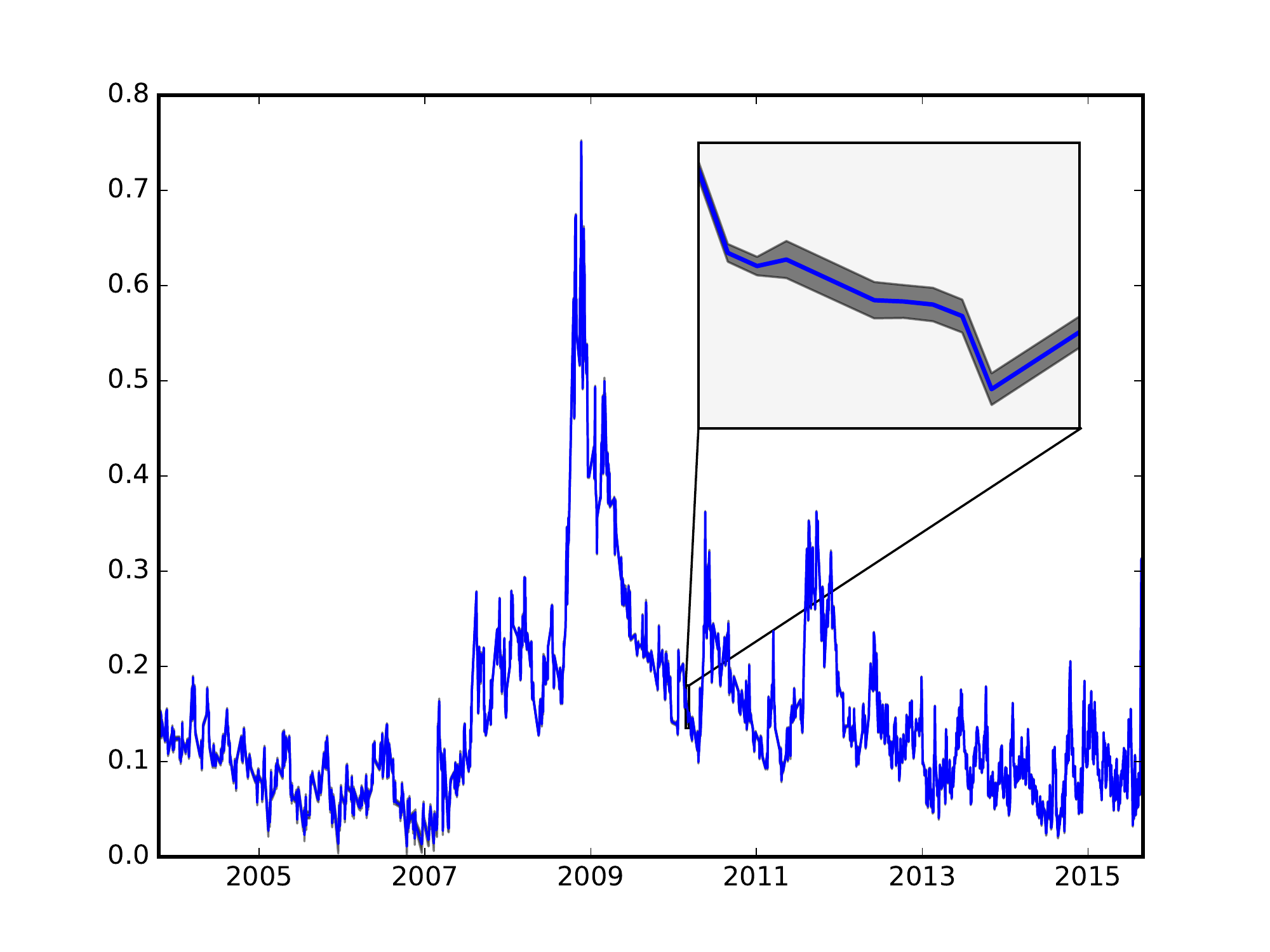}
\caption{\label{fig:instVol} Volatility of the S{\&}P 500. The close-up presents the
volatility estimate along with its double standard deviation obtained from Fisher information matrix \eq{fisherVIX}.}
\end{figure}
\fig{instVol} presents the volatility time-series $(\sg_t)_{t \in \mathcal{T}}$ obtained from the previously described fit on components entering the calculation of the VIX along with the uncertainty left. \\
Since the uncertainty left is negligible the shaded tube  $([\sg_t - \b_t, \sg_t + \b_t])_{t\in \mathcal{T}}$ (95{\%} credibility region) mantling the plot of the volatility time-series $(\sg_t)_{t \in \mathcal{T}}$ is only visible in a close-up presented in \fig{instVol} as well. The close-up shows volatility from February 23rd 2010 until March 10th 2010 within a range from $0.135$ till $0.180$. The boundaries $(\b_t)_{t \in \mathcal{T}}$ of the credibility region are obtained from the first $m$ entries of the diagonal of the inverse $J(\vec \Theta)^{-1} = (J(\vec \Theta)^{-1}_{ij})$ of the Fisher information matrix (Recall, $m$ is the number of trading days we consider, i.e. the cardinality of the set $\mathcal{T}$):
\[
\b_{t_i}  = 2 \sqrt{J(\vec \Theta)^{-1}_{ii}} \quad \text{for all } i = 1, \ldots, m \, .
\]
That is, in terms of the Cram\'{e}r-Rao inequality, $\b_t$ represents a lower bound on the double standard deviation of the volatility estimate $\sg_t$. On average the double standard
deviation is
\[
\bar{\b} = \dfrac{1}{m} \sum_{i=0}^{m} \b_{t_i} = 0.0014 \, .
\]
The last four entries of the diagonal of $J(\vec \Theta)^{-1} = (J(\vec \Theta)^{-1}_{ij})$ yield a lower bound on the variance of the estimates of the parameters $\k, \sqrt{\t}, \sg$ and $\rho$, respectively. Along with the parameter estimates of \cite{Ait-Sahalia2007}	 we obtain table \ref{ErrTab}. Fitting the parameters $\k, \sqrt{\t}, \sg$ and $\rho$ on options over a sufficiently long time-window yields fairly accurate estimates of them as well.\\

\begin{table}[ht]
\begin{center}
\begin{tabular}{ccc}
 & Estimate & Standard Error \\ 
$\k$ & $5.07$ & $4.0e-2$ \\ 
$\sqrt{\t}$ &  $0.214$ & $6.5e-4$ \\ 
$\sg$ & $0.48$ & $9.4e-4$ \\ 
$\rho$ & $-0.767$ & $7.5e-4$ \\ 
\end{tabular}
\caption{ \label{ErrTab} The parameters of Heston's model along with their standard errors obtained from the Fisher information matrix if we assume they were estimated from the components of the VIX between October 16th 2003 until August 31th 2015.}
\end{center}
\end{table} 
Overall, estimates of hidden volatility and the parameters determining the stochastic process \eq{CoxIngersoll} from option data appear reliable. Doubts are shed on these results if relative errors are considered instead of absolute ones. \fig{relErr} presents the relative uncertainty left, that is, the time-series $(\b_t/\sg_t)_{t \in \mathcal{T}}$.
\begin{figure}[ht]
\includegraphics[width=1.0\linewidth]{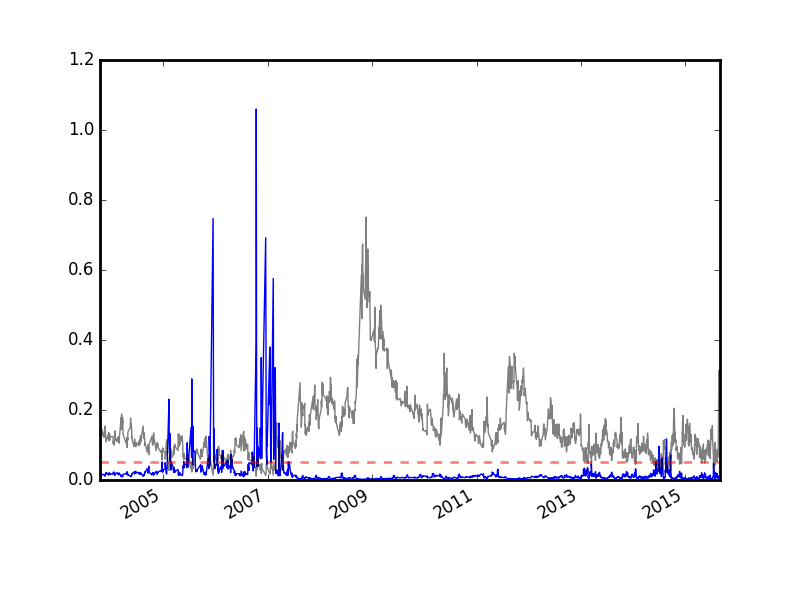}
\caption{\label{fig:relErr} The relative error of the volatility estimate (blue line) and
volatility itself (transparent black line). If volatility drops below a certain
level (dashed, transparent red line), option data yield no information about volatility.}
\end{figure}
Apparently, the uncertainty becomes overwhelming, if the volatility drops. This is in accordance with the discussion in the previous section where we observed that Vega drops if volatility falls below a critical value. From the block structure \eq{fisherVIX} of the Fisher information matrix and the proof of lemma \ref{singleFit} follows
\[
J(\vec \Theta)^{-1}_{ii} \geq \dfrac{\hat v}{(\partial_{\sg_{t_i}} E_i)^2}  
\]
with $E_i$ provided by \eq{Ei}, and therefore
\[
\b_{t_i} \geq \dfrac{2 \sqrt{\hat v}}{\partial_{\sg_{t_i}} E_i} 
\]
for all $i = 1, \ldots, m$. Hence, there is a lower bound of the error of the  volatility $\sg_t$ proportional to the inverse of a sum of Vegas. If this sum shrinks the error becomes large.

\subsection{The VIX as variance swap}

The VIX is quoted in percentage points and translates, roughly, to the expected 
movement (with the assumption of a 68{\%} likelihood, i.e., one standard deviation) in 
the S{\&}P 500 index over the next 30-day period, which is then annualized. According to Carr and Wu \cite{Carr2004} the VIX can also be considered as $30$ day variance swap on the S{\&}P 500. Recall \eq{CoxIngersoll} that the variance of the Heston model is driven by the CIR process
\[
dv_t = \k(\t - v_t) dt + \sg \sqrt{v_t} dW_{t} 
\]
and consequently, that the expected value of $v_t$ conditional on $v_s$ $(s < t)$ is
\[
\mathbb{E}[v_t | v_s ] = v_s e^{-\k(t-s)} + \theta \left(1 - e^{-\k(t-s)} \right) = (v_s - \t) e^{-\k(t-s)} + \t \, .
\] 
In the sequel, we make use of $\mathbb{E}[v_t | v_s ]$ but with $s = 0$. It is useful to denote this quantity as $\hat v_t$
\[
\hat v_t = \mathbb{E}[v_t | v_0 ] = (v_0 - \t) e^{-\k t} + \t \, .
\]
It is also useful to define the total (integrated) variance $\hat w_t$ as 
\[
\hat w_t = \int_{0}^{t} v_s \, ds = (v_0 - \t) \dfrac{1 - e^{-\k t}}{\k}	+ \t t \, .
\]
As explained by Gatheral \cite{GatheralVolSur}, a variance swap requires an estimate of the future variance over the $(0, T)$ time period. This can be obtained as the conditional expectation of the integrated variance. A fair estimate of the total variance is therefore
\begin{align*}
\mathbb{E} \left[ \int_0^{T} v_t dt | v_0 \right] &= \int_0^{T} \mathbb{E} [v_t|v_0] \, dt \\
&=  \int_0^{T} (v_0 - \t) e^{-\k t} + \t \, dt = (v_0 - \t) \dfrac{1 - e^{-\k T}}{\k}	+ \t T
\end{align*}
which is simply $\hat w_T$. Since this represents the total variance over $(0, T)$, it must be scaled by $T$ in order to represent a fair estimate of annual variance (assuming that $T$ is expressed in years). Hence, the strike variance $K_{\text{var}}^{2}$ for a variance swap is obtained by dividing this last expression by $T$
\[
K_{\text{var}}^{2} = (v_0 - \t) \dfrac{1 - e^{-\k T}}{\k T} 	+ \t \, .
\]
Returning to Carr's and Wu's interpretation of the VIX \cite{Carr2004} the VIX-time series $(\text{VIX}_t)_{t \in \mathcal T}$ is the time-series 
\[
K_{\text{var}, t} = \sqrt{(v_t - \t) \dfrac{1 - e^{-\k T}}{\k T} 	+ \t} \quad \text{for } t \in \mathcal T
\]
where $(v_t)_{t \in \mathcal{T}}$ denotes the  variance of the S{\&}P 500 at day $t$ and $T = 30/365$. Thus, Fisher information provides an uncertainty estimate on the VIX, considered as $30$ days variance swap, estimated from its components, i.e., near- and next-term put and call options with more than 23 days and less than 37 days to expiration and non-vanishing bid. We only have to transform the Fisher information matrix already computed in the previous subsection according to the rules \eq{fisherTrafo}. In terms of \eq{fisherTrafo} in section 2.1 we have
\begin{align*}
\vec \Lambda &= (K_{\text{var}, t})_{t \in \mathcal{T}} \oplus (\k, \sqrt{\t}, \sg, \rho) \\
\vec \Theta &= (\sg_t)_{t \in \mathcal{T}}\oplus (\k, \sqrt{\t}, \sg, \rho) \, 
\end{align*}
and 
\[
J(\vec \Lambda) = D^T J(\vec \Theta) D \, .
\]
Hence, the transformation matrix $D$ in \eq{fisherTrafo} adopts the block form
\[
D = \left( \begin{array}{cc}
\vec A_{11} & \vec A_{12} \\ 
\vec A_{21} & \vec A_{22}
\end{array} \right) \, .
\]
$\vec A_{11}$ is an $m \times m$ diagonal matrix with entries
\begin{align*}
a_{ii} &= \dfrac{\partial \sg_{t_i}}{\partial K_{\text{var}, t_i}} = \dfrac{\partial \sg_{t_i}}{\partial v_{t_i}}\dfrac{\partial v_{t_i}}{\partial K_{\text{var}, t_i}} \\
&= \dfrac{1}{2 \sg_{t_i}} \dfrac{\partial }{\partial K_{\text{var}, t_i}} \left(  \dfrac{(K_{\text{var}, t_i}^{2} - \t) \k T }{1 - e^{-\k T}} + \t \right) 
= \dfrac{1}{\sg_{t_i}} \dfrac{K_{\text{var}, t_i} \k T}{1 - e^{-\k T}}
\end{align*}
for $i = 1, \ldots, m$ where $m$ is the number of trading days $t_i \in \mathcal{T}$  from October 16th 2003 until August 31th 2015. $\vec A_{12}$ is an $m \times 4$-matrix with entries
\[
\left( \begin{array}{cccc}
\dfrac{\partial \sg_{t_0}}{ \partial \k} & \dfrac{\partial \sg_{t_0}}{ \partial \sqrt{\t}}  & 0 & 0 \\ 
\vdots & \vdots & \vdots & \vdots \\ 
\dfrac{\partial \sg_{t_m}}{ \partial \k} & \dfrac{\partial \sg_{t_m}}{ \partial \sqrt{\t}} & 0 & 0
\end{array} \right)
\]
where
\begin{align*}
\dfrac{\partial \sg_{t_i}}{ \partial \k} &= \dfrac{\partial \sg_{t_i}}{\partial v_{t_i}}\dfrac{\partial v_{t_i}}{\partial \k} = \dfrac{(K_{\text{var}, t_i}^2 - \t)T}{2 \sg_{t_i}} \dfrac{\partial}{\partial \k} \left( \dfrac{\k}{1 - e^{-\k T}} \right) \\
&= \dfrac{(K_{\text{var}, t_i}^2 - \t)T}{2 \sg_{t_i}} \dfrac{1 - (1 + \k T)e^{-\k T}}{\left(1 - e^{-\k T}\right)^2}\\
\dfrac{\partial \sg_{t_i}}{ \partial \sqrt{\t}} &= \dfrac{\partial \sg_{t_i}}{\partial v_{t_i}} \dfrac{\partial \t}{\partial \sqrt{\t}} \dfrac{\partial v_{t_i}}{\partial \t} = \dfrac{\sqrt{\t}}{\sg_{t_i}} \left( 1 - \dfrac{\k T}{1 - e^{-\k T}} \right) \, .
\end{align*}
Finally, $\vec A_{21} = \vec A_{12}^T$ and $\vec A_{22} $ is the $4 \times 4$ identity matrix. Similar to figure \fig{instVol} we plot \fig{vixTube} the time-series $(K_{\text{var}, t})_{t \in \mathcal{T}}$ along with its $95{\%}$ credibility region $[K_{\text{var}, t}- \b_t, K_{\text{var}, t} + \b_t]$ where
\[
\b_{t_i}  = 2 \sqrt{J(\vec \Lambda)^{-1}_{ii}} \quad \text{for } i = 1, \ldots, m \, .
\]
Furthermore, the historical VIX is plotted. The value $K_{\text{var}, t}$ is systematically smaller than the realized VIX. 

\begin{figure}[h]
\includegraphics[width=1.0\linewidth]{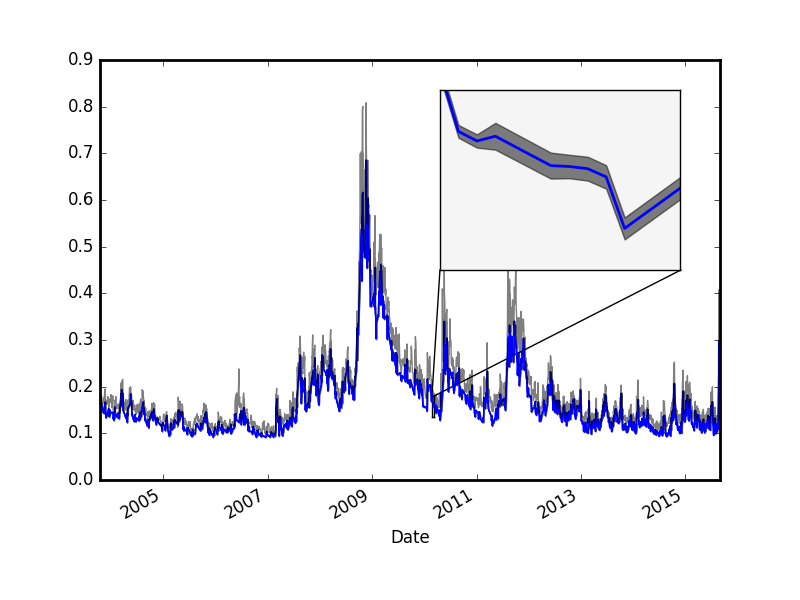}
\caption{\label{fig:vixTube} The time series $(K_{\text{var}, t})_{t \in \mathcal{T}}$. The close-up shows the graph along with its $95{\%}$ credibility region. The Black transparent line shows the historical VIX.}
\end{figure}
\begin{figure}[h]
\includegraphics[width=1.0\linewidth]{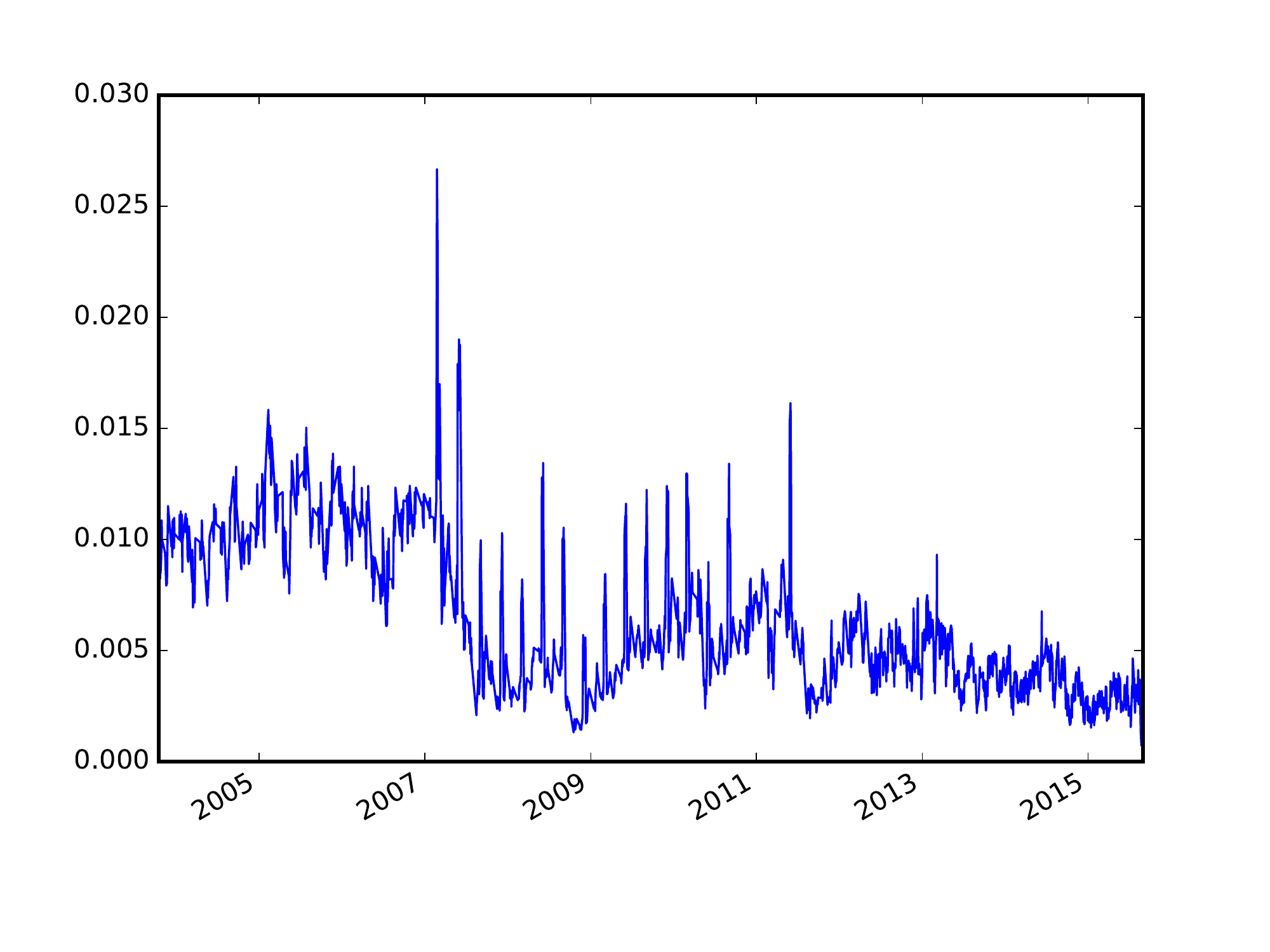}
\caption{\label{fig:vixRel} The time series $(\b_t / K_{\text{var}, t})_{t \in \mathcal{T}}$.}
\end{figure}
As for the  volatility the error is negligible and on average we have
\[
\bar{\b} = \dfrac{1}{m} \sum_{i=0}^{m} \b_{t_i} = 8.9e-04 \, .
\]
Compared to the  volatility estimates from options there is big difference if we consider the relative error, that is, the time-series $\b_t / K_{\text{var}, t}$ for $t \in \mathcal{T}$ in \fig{vixRel}. The relative error never exceeds $3 {\%}$.

\section{Conclusion}

Here, we have addressed the question of how reliably
volatility can be estimated from option price data. To this end, we
computed the Fisher information matrix of Heston's stochastic
volatility model. Thanks to the analytic tractability of Heston's
model, the Fisher information can be expressed in Fourier integrals
giving the Heston Greeks.

Our investigations lead to the following insights: First, options at
the money are most informative about  volatility while
almost no information can be obtained from options that are far out of
the money. Second, low  volatilities are hard to estimate
as Vega almost vanishes in this case, making it impossible to extract
information from option prices. Third, volatility estimation from
market data as exemplified on S\&P 500 index options is reliable most
of the time, with occasional large relative errors for very low
volatilities. We might speculate that this could lead to overconfident
estimates of protfolio risk especially in times of calm financial
markets.  Nevertheless, the VIX index itself, reflecting the average
volatility over the next month, proves to be an accurate accessment of
 volatility.

Overall, this work complements our previous findings \cite{PfanBerti2016} regarding the low
information content of stock returns. There, we showed that in general at least secondly
quoted return data is necessary to infer volatility successfully. We guessed that option price data
could lead to much more reliable volatility estimates as confirmed by the analysis presented here.

\section{Funding}

Nils Bertschinger and Oliver Pfante thank Dr. h.c. Maucher for funding their positions.

\bibliographystyle{abbrv}
\bibliography{Heston}
\end{document}